\documentclass{amsart}

\usepackage{amsthm}
\usepackage{amsmath} \usepackage{amsfonts}
\usepackage{amssymb} \usepackage{latexsym} \usepackage{enumerate}
\usepackage{mathtools}

\usepackage{mathrsfs}
\usepackage[numbers]{natbib}

\numberwithin{equation}{section}
\newtheorem{theorem}{Theorem}[section]
\newtheorem{lemma}[theorem]{Lemma}
\newtheorem{remark}[theorem]{Remark}
\newtheorem{corollary}[theorem]{Corollary}
\newtheorem{proposition}[theorem]{Proposition}

\begin{document}
\title{Optimal Liquidation with High Risk Aversion and Small Linear Price Impact}
  \thanks{YD Supported in part by the GIF Grant 1489-304.6/2019 and the ISF grant 230/21}

   \author{Leonid Dolinskyi \address{
 Department of Finance, National University of Kyiv-Mohyla Academy.\\
 e.mail: ldolinskyi@ukma.edu.ua}} 
  \author{Yan Dolinsky \address{
 Department of Statistics, Hebrew University of Jerusalem.\\
 e.mail: yan.dolinsky@mail.huji.ac.il}}
  \date{\today}
\maketitle
\begin{abstract}
We consider the Bachelier model with linear price impact. Exponential utility
indifference prices are studied for vanilla European options
in the case where the investor is required to liquidate her position.
Our main result is establishing a
non-trivial scaling
limit for a vanishing price impact which is inversely proportional to the risk aversion. We compute the limit of the corresponding utility indifference prices and find explicitly a family of portfolios which are asymptotically optimal.
\end{abstract}
\begin{description}
\item[Mathematical Subject Classification (2010)] 91B16, 91G10, 60H30
\item[Keywords] exponential utility, linear price impact, optimal liquidation
\end{description}

 \maketitle \markboth{}{}
\renewcommand{\theequation}{\arabic{section}.\arabic{equation}}
\pagenumbering{arabic}

\section{Introduction}\label{sec:1}
In financial markets, trading moves prices against the trader:
buying faster increases execution prices, and selling faster decreases them. This aspect
of liquidity, known as market depth (see Black (1986))
 or price-impact, has received large
attention in optimal liquidation problems,
see, for instance,
Almgren \& Chriss (2001), Schied et al. (2010), Gatheral \& Schied (2011),
Bayrakatar \& Ludkovski (2014),
Bank \& Voß (2019), Fruth et al. (2019), 
and the references therein.

It is well known that in 
the presence of
price impact, super–replication is prohibitively costly, see  Guasoni \& Rasonyi (2015). 
Namely, in the presence of price
impact, even in market models such as the Bachelier model or the Black–Scholes model
(which are complete in the frictionless setup) there is no practical way to construct a
hedging strategy which eliminates all risk from a financial position. 
This brings us to introducing 
preferences.
We assume that the preferences
of the agent are given by an exponential utility. 
Then, the
optimal hedging strategy is determined by 
maximizing the expected exponential utility of the terminal wealth
generated by the dynamic trading in the underlying asset minus 
the liability of the investor which is equal to
 the payoff of the option. 
A natural notion of option
pricing, in this setting, is the utility 
indifference price (which we define in Section 2).

In this paper we consider the problem of optimal liquidation for the exponential
utility function in a 
model with temporary linear price impact. 
Formally, we study exponential utility maximization in the presence of quadratic transaction costs and the
constraints that the number of shares at the maturity date is zero.
The motivation for the later constraint is that in real market conditions 
many of the derivative securities (such as European options) are cash settled.  

We compute the asymptotic behavior of the exponential utility indifference prices where
the risk aversion goes to infinity at a rate which is inversely proportional to the linear
price impact which goes to zero. In addition we provide a family of asymptotically
optimal hedging strategies. 
We divide the proof of our main result (Theorem \ref{thm.1}) into two main steps:
the proof of the lower bound and the proof of the upper bound.
In the proof of the lower bound we apply
 Theorem 2.2 from Dolinsky (2022) which
gives a dual representation of the certainty
equivalent for the case where the investor has to liquidate her position. This dual representation together with the Brownian structure
allows us to compute the scaling limit of the utility indifference prices.
The proof of the upper bound is done by an explicit construction of a family of portfolios which are asymptotically optimal.

The above type of scaling limits goes back to 
the seminal work of Barles \& Soner (1998) which
determines the scaling limit of utility indifference prices of vanilla options
for small proportional transaction costs and high risk aversion. 
The present
paper provides an analogous analysis for the case of quadratic transaction costs, albeit
using convex duality and martingale techniques rather than taking a PDE
approach as pursued in Barles \& Soner (1998). 
The financial idea behind this approach is to produce reasonable 
option prices which in the presence of small friction allow 
to "almost" super--hedge the derivative security (this corresponds to large risk aversion). 

The current work is also closely related to the recent paper 
Ekren \& Nadtochiy (2022)
where the authors considered utility--based hedging
with quadratic transaction costs and Bachelier dynamics for the unaffected stock price. 
For a given risk aversion the authors apply the Hamilton–Jacobi-Bellman (HJB) 
methodology and obtain a representation of the value function and
the optimal strategy. For technical reasons, instead of requiring that the 
number of shares at the maturity date be zero, 
they penalize the square of the number of shares at the maturity date. 

The rest of the paper is organized as follows. In the next section we introduce the
setup and formulate the main results.
In Section 3
we prove the lower bound.
In Section 4 we prove the upper bound. In Section 5 we derive an
auxiliary result from the field of deterministic variational analysis.

\section{Preliminaries and Main Results}
Let $T<\infty$ be the time horizon and let
$W=(W_t)_{t \in [0,T]}$
be a standard one dimensional Brownian motion defined on the filtered probability space
$(\Omega, \mathcal{F},(\mathcal F_t)_{t\in [0,T]}, \mathbb P)$
where the filtration $(\mathcal F_t)_{t\in [0,T]}$ satisfies the
usual assumptions (right continuity and completeness).
We consider a simple financial market
with a riskless savings account bearing zero interest (for simplicity) and with a risky
asset $S=\left(S_t\right)_{t \in [0,T]}$ with Bachelier price dynamics
\begin{equation}\label{2.bac}
S_t=S_0+\mu t+\sigma W_t
\end{equation}
where $S_0\in \mathbb R$ is the initial position of the risky asset,
$\mu\in\mathbb R$ is the constant drift and
$\sigma>0$ is the constant volatility.

Following Almgren \& Chriss (2001),
we model the investor’s market impact, in a temporary linear form
and, thus, when at time $t$ the investor turns over her position $\Phi_t$ at the rate $\dot{\Phi}_t:=\frac{d\Phi_t}{dt}$
the execution price is $S_t+\frac{\Lambda}{2}\dot{\Phi}_t$ for some constant $\Lambda>0$.
The portfolio
value at the maturity date is given by
\begin{equation}\label{2.1}
V^{\Phi}_T:=\int_{0}^T \Phi_t dS_t-\frac{\Lambda}{2}\int_{0}^T \dot{\Phi}^2_t dt.
\end{equation}

In our setup the investor has to liquidate
her position. Thus, the natural class of admissible strategies which we denote by $\mathcal A$ is
the set of all progressively measurable processes $\Phi=(\Phi_t)_{t\in [0,T]}$ with
differentiable trajectories
such that $\int_{0}^T \dot{\Phi}^2_t dt<\infty$ and $\Phi_T=0$ almost surely.
 We assume that the initial number of shares $\Phi_0$ is fixed.

Consider a vanilla European option with the payoff $X=f(S_T)$ where $f$ is of the form
\begin{equation}\label{2.form}
f(x)=\max\left(0, \Theta \left(x-K\right)\right), \ \ x\in\mathbb R
\end{equation}
for some constants $\Theta,K\in \mathbb R$. Observe that this form includes call/put options.

The investor will assess the quality of a hedge by the resulting expected utility.
Namely, we follow the approach proposed in Hodges \& Neuberger (1989) which says
that the price of the contingent claim (from the seller's point of view) is 
the amount leading the investor to be indifferent between the following two actions:\\
(i) Selling the option and hedging it. (ii) Hedging with no option. 
Thus, assuming exponential utility with constant absolute risk aversion $\alpha>0$, the utility
indifference price 
$\pi=\pi(\Lambda,\alpha,\Phi_0,X)$
satisfies
$$\sup_{\Phi\in\mathcal A}\mathbb E_{\mathbb P}\left[-\exp\left(\alpha\left(X-\pi-V^{\Phi}_T\right)\right)\right]=
\sup_{\Phi\in\mathcal A}\mathbb E_{\mathbb P}\left[-\exp\left(-\alpha V^{\Phi}_T\right)\right]
.$$
We obtain the definition
\begin{equation}\label{2.2}
\pi(\Lambda,\alpha,\Phi_0,X):=
\frac{1}{\alpha}\log\left(\frac{\inf_{\Phi\in\mathcal A}\mathbb E_{\mathbb P}\left[\exp\left(\alpha\left(X-V^{\Phi}_T\right)\right)\right]}
{\inf_{\Phi\in\mathcal A}\mathbb E_{\mathbb P}\left[\exp\left(-\alpha V^{\Phi}_T\right)\right]}\right).
\end{equation}

The certainty equivalent
of the claim $X$ is given by 
$$c(\Lambda,\alpha,\Phi_0,X):=
\frac{1}{\alpha}\log\left(\inf_{\Phi\in\mathcal A}\mathbb E_{\mathbb P}\left[\exp\left(\alpha\left(X-V^{\Phi}_T\right)\right)\right]\right).$$
Namely, $c:=c(\Lambda,\alpha,\Phi_0,X)$ satisfies 
$$\sup_{\Phi\in\mathcal A}\mathbb E_{\mathbb P}\left[-\exp\left(\alpha\left(X-c-V^{\Phi}_T\right)\right)\right]=-1.$$
Economically speaking, the term $c$ is the amount leading the investor to be indifferent between 
the following: (i) Selling the option and hedging it. (ii) Doing nothing  $(-1=-e^{-\alpha 0})$.
From the economics point of view, the certainty equivalent 
is a more appropriate term for the buyer of the option. Hence, in our setup (we treat the seller) 
we are mainly interested in the utility indifference price 
$\pi(\Lambda,\alpha,\Phi_0,X)$. The certainty equivalent $c(\Lambda,\alpha,\Phi_0,X)$ 
can be viewed as the logarithmic scale 
for the 
value of the utility maximization problem 
$\sup_{\Phi\in\mathcal A}\mathbb E_{\mathbb P}\left[-\exp\left(\alpha\left(X-V^{\Phi}_T\right)\right)\right]$. Moreover, the certainty equivalent term
appears naturally in the 
dual representation from Dolinsky (2022).  
For more details on utility indifference pricing see
Carmona (2009).

We notice that if the risk aversion $\alpha>0$ is fixed, then by applying standard density arguments we obtain that
for $\Lambda\downarrow 0$, the above indifference price converges to the unique price of the continuous
time complete (frictionless) market given by (\ref{2.bac}). A more interesting limit emerges,
however, if we re-scale the investor’s risk-aversion in the form $\alpha:=A/\Lambda$.
Hence, we fix $A>0$ and consider the case where the risk aversion is  $\alpha(\Lambda):=\frac{A}{\Lambda}$.

A simple and rough intuition for this type of scaling is done by applying 
Schied \& Sch\"{o}neborn (2007). Indeed, consider the simple case where $\sigma=1$ and $\mu=0$. Then, from Theorem 2.1 in 
Schied \& Sch\"{o}neborn (2007) we have 
$$\inf_{\Phi\in\mathcal A}\mathbb E_{\mathbb P}\left[\exp\left(-\alpha V^{\Phi}_T\right)\right]=
\exp\left(-\alpha\Phi_0 S_0+\frac{\alpha \Phi^2_0}{2}\sqrt{\alpha\Lambda}\coth\left(\sqrt\frac{\alpha}{\Lambda}T\right)\right).$$
Hence,  
if we are looking for a scaling such that 
$$\frac{1}{\alpha}\log\left(\inf_{\Phi\in\mathcal A}\mathbb E_{\mathbb P}\left[\exp\left(-\alpha V^{\Phi}_T\right)\right]\right)=
-\Phi_0 S_0+\frac{ \Phi^2_0}{2}\sqrt{\alpha\Lambda}\coth\left(\sqrt\frac{\alpha}{\Lambda}T\right)$$
will converge  as $\Lambda\downarrow 0$ and $\alpha\rightarrow\infty$ (for any $\Phi_0$), 
then the right scaling is $\alpha(\Lambda):=\frac{A}{\Lambda}$.

Before we
formulate the main result we need some preparations.
Introduce the functions
\begin{equation}\label{def1}
g(x):=\sup_{y\in\mathbb R} \left[f(x+y)-\frac{y^2}{4\sigma\sqrt A}\right]=
\max\left(0, \Theta\left(x-K\right)+\sigma\sqrt A \Theta^2 \right), \ \  x\in \mathbb R
\end{equation}
and
$$
u(t,x):=\mathbb E_{\mathbb P}\left[g(x+\sigma W_{T-t} )\right], \quad (t,x)\in [0,T]\times\mathbb R.
$$
The term $u(t,S_t)$
represents the price at time $t$ of a European option with the payoff $g(S_T)$
in the complete market given by (\ref{2.bac}).
It is well known that $u\in C^{1,2}([0,T)\times\mathbb R)$ solves the PDE
\begin{equation}\label{PDE}
\frac{\partial u}{\partial t}+\frac{\sigma^2}{2}\frac{\partial^2 u}{\partial x^2}=0 \ \ \ \ \mbox{in} \ \ [0,T)\times\mathbb R.
\end{equation}
Next, let $\Lambda>0$ and let
$$\rho=\rho(\Lambda):=\frac{\sigma^2 \alpha(\Lambda)}{\Lambda}=\frac{\sigma^2 A }{\Lambda^2}$$
 be the risk-liquidity ratio.
Consider the (random) ODE on the interval $[0,T]$
\begin{eqnarray}\label{ODE}
&\dot{F}_t=\sqrt\rho\left(\frac{\cosh(\sqrt\rho (T-t))}{2 \cosh^2\left(\frac{\sqrt\rho (T-t)}{2}\right)}\frac{\partial u}{\partial x}(t,S_t-\sigma\sqrt A F_t)-\tanh(\sqrt\rho (T-t))F_t\right), \\
&\mbox{with} \ \mbox{the} \ \mbox{initial} \ \mbox{condition} \ \ F_0=\Phi_0\coth(\sqrt\rho T)\nonumber.
\end{eqnarray}
From the linear growth of $g$ it follows that for any $\epsilon>0$ the functions
$\frac{\partial u}{\partial x},\frac{\partial^2 u}{\partial x^2}$ are uniformly bounded in the domain $[0,T-\epsilon]\times\mathbb R$. In particular
$\frac{\partial u}{\partial x}$ is Lipschitz continuous with respect to $x$ in the domain $[0,T-\epsilon]\times\mathbb R$.
Observe that the functions
$\frac{\cosh(\sqrt\rho (T-t))}{2 \cosh^2\left(\frac{\sqrt\rho (T-t)}{2}\right)}, \tanh(\sqrt\rho (T-t))$ are bounded.
Hence, from the standard theory of
ODE
(see
Walter (1998), Chapter II, Section 6) we obtain that
there exists a unique solution to (\ref{ODE})
which we denote by $F^{\Lambda}=(F^{\Lambda}_t)_{t\in [0,T)}$ and the solution is
Lipschitz continuous, and so
 $\lim_{t\rightarrow T-}F^{\Lambda}_t$ exists. Set $F^{\Lambda}_T:=\lim_{t\rightarrow T-}F^{\Lambda}_t$
 and define
\begin{equation}\label{2.3}
\Phi^{\Lambda}_t:=\tanh\left(\sqrt{\rho(\Lambda)} (T-t)\right)   F^{\Lambda}_t, \ \ t\in [0,T].
\end{equation}

\begin{theorem}\label{thm.1}
For vanishing linear price impact $\Lambda \downarrow 0$ and re-scaled high risk-aversion
$A/\Lambda$ with $A>0$ fixed, the certainty equivalent of
$X=\max\left(0, \Theta \left(S_T-K\right)\right)$
has the scaling limit
\begin{equation}\label{1}
\lim_{\Lambda\downarrow 0} c(\Lambda,A/\Lambda,\Phi_0,X)=
u\left(0,S_0-\sigma\sqrt A\Phi_0\right)+\frac{\sigma\sqrt A \Phi^2_0 }{2}.
\end{equation}
  Moreover, the trading strategies given by (\ref{2.3}) are asymptotically optimal, i.e.
\begin{equation}\label{2}
 \lim_{\Lambda\downarrow 0}\frac{\Lambda}{A}
 \log\left(\mathbb E_{\mathbb P}\left[\exp\left(\frac{A}{\Lambda}\left(X-V^{\Phi^{\Lambda}}_T\right)\right)\right]
\right)=u\left(0,S_0-\sigma\sqrt A\Phi_0\right)+\frac{\sigma\sqrt A \Phi^2_0 }{2}.
\end{equation}
\end{theorem}
From Theorem \ref{thm.1} we obtain immediately the following corollary which says that the
asymptotic value of the utility indifference prices is equal to the price of the vanilla European option with the payoff $g(S_T)$ and the shifted initial stock price
$S_0-\sigma \sqrt A\Phi_0$.
\begin{corollary}
For vanishing linear price impact $\Lambda \downarrow 0$ and re-scaled high risk-aversion
$A/\Lambda$ with $A>0$ fixed, the utility indifference price of $X$ given by (\ref{2.2}) has the scaling limit
$$
 \lim_{\Lambda\downarrow 0}\pi(\Lambda,A/\Lambda,\Phi_0,X)=u\left(0,S_0-\sigma\sqrt A\Phi_0\right).$$
\end{corollary}
\begin{proof}
Apply (\ref{1}) and take $X\equiv 0$ for the
denominator of (\ref{2.2}).
\end{proof}
\begin{remark}\label{rem0}
In the proof of the lower bound (given in the next section) we only assume that the payoff function $f$
is Lipschitz continuous. By a more careful analysis we can prove that in fact there is an equality, namely (\ref{1}) holds true for any payoff function $X=f(S_T)$ with a Lipschitz continuous $f$.
 Unfortunately,
 the proof of (\ref{2}) (given in Section 4) uses the specific structure of the payoff given by (\ref{2.form}).
 This together with the fact that the most common vanilla options in real markets are of the form (\ref{2.form})
 led us to assume from the beginning that the payoff is of this form.

 Let us emphasize that our results can be extended to the multi--asset case with a similar proof. In the multi asset case
 the volatility $\sigma$ is replaced with a positive definite matrix
 and the functions $\coth$ and $\tanh$ are viewed as matrix valued functions.
 \end{remark}
 \begin{remark}\label{rem1}
 The current setup without the liquidation requirement $\Phi_T=0$ was studied in 
 Dolinsky \& Moshe (2022). 
In both cases (with or without liquidation) the scaling limit of the utility indifference prices
(with the same scaling $\alpha(\Lambda)=\frac{A}{\Lambda}$)
 is equal to $\mathbb E_{\mathbb P} \left[h\left(S_0-\sigma\sqrt A\Phi_0+\sigma W_T\right)\right]$ 
 for a modified function $h$. 
 In the present paper
 $$h(x):=g(x)=\sup_{y\in\mathbb R} \left[f(x+y)-\frac{y^2}{4\sigma\sqrt A}\right]=\max\left(0, \Theta\left(x-K\right)+\sigma\sqrt A \Theta^2 \right)$$
 while in Dolinsky \& Moshe (2022) the modified payoff is smaller and given by
 $$h(x):=\sup_{y\in\mathbb R} \left[f(x+y)-\frac{y^2}{2\sigma\sqrt A}\right]=\max\left(0, \Theta\left(x-K\right)+\frac{\sigma\sqrt A \Theta^2}{2} \right).$$
 In both cases the function $h$ is strictly larger (provided that $A>0$) than the original payoff function $f$.
The term  $\mathbb E_{\mathbb P} \left[h\left(S_0-\sigma\sqrt A\Phi_0+\sigma W_T\right)\right]$ represents the option price of the modified claim 
$h(S_T)$ in the complete (frictionless) Bachelier model with volatility 
$\sigma$ and shifted initial stock price $S_0-\sigma\sqrt A\Phi_0$.
 
Observe that if the risk aversion $\alpha$ is constant and $\Lambda=0$ (i.e. no friction)
 then formally $A=\alpha\Lambda=0$ and so, in this case the function $h$ coincides with the original payoff $f$
 and there is no shift in the initial stock price $S_0$.
Namely, 
 we recover 
 the price for the complete (frictionless) Bachelier model given by (\ref{2.bac}).  

Next, we discuss briefly the difference between the asymptotically optimal portfolios which are given by Theorem \ref{thm.1}
and those given in Dolinsky \& Moshe (2022).
From (\ref{2.3}) we have
\begin{equation}\label{4.0+}
\dot{\Phi}^{\Lambda}_t=\sqrt{\rho(\Lambda)}\left(\tanh\left(\frac{\sqrt{\rho(\Lambda)}\left(T-t\right)}{2} \right)\Upsilon^{\Lambda}_t-
\coth\left(\sqrt{\rho(\Lambda)} \left(T-t\right)\right)\Phi^{\Lambda}_t\right)
\end{equation}
where $\rho(\Lambda):=\frac{\sigma^2 A}{\Lambda^2}$ and  $\Upsilon^{\Lambda}_t:=\frac{\partial u}{\partial x}(t,S_t-\sigma\sqrt A F^{\Lambda}_t)$,  $t\in [0,T)$.

Thus, we have a mean reverting structure which combines tracking the
$\Delta$--hedging strategy $(\Upsilon^{\Lambda}_t)_{t\in [0,T]}$ of a modified claim $g$ and liquidating the position
up to the maturity date. 
As time $t$ approaches maturity the weight $\sqrt{\rho(\Lambda)}\tanh\left(\frac{\sqrt{\rho(\Lambda)}\left(T-t\right)}{2} \right)$of the $\Delta$--hedging strategy 
vanishes  and 
due to the term $\sqrt{\rho(\Lambda)}\coth\left(\sqrt{\rho(\Lambda)} \left(T-t\right)\right)$ (goes to $\infty$ for $t\uparrow T$)
the investor
trading is mainly towards liquidation.
This is in contrast to the asymptotically optimal portfolios in Dolinsky \& Moshe (2022)
which are just based on tracking the appropriate $\Delta$--hedging strategy.

In broad terms the methods of the proof in this paper are close to those in Dolinsky \& Moshe (2022)
and 
based on duality and explicit construction of 
asymptotically optimal portfolios. However, the additional constraint that the number of shares
at the maturity date is zero (i.e. liquidation), makes the mathematical analysis more challenging. 
In particular it requires a dual representation 
which was obtained recently in Dolinsky (2022) and treats the liquidation case.   
\end{remark}
\section{Proof of the Lower Bound}
In this section we prove the following statement.
\begin{proposition}\label{prop3.1}
For vanishing linear price impact $\Lambda \downarrow 0$ and re-scaled high risk-aversion
$A/\Lambda$ with $A>0$ fixed, we have the following lower bound
\begin{equation*}\label{3}
\lim\inf_{\Lambda\downarrow 0} c(\Lambda,A/\Lambda,\Phi_0,X)\geq
u\left(0,S_0-\sigma\sqrt A\Phi_0\right)+\frac{\sigma\sqrt A \Phi^2_0 }{2}.
\end{equation*}
\end{proposition}
We start with the following Lemma.
\begin{lemma}\label{lem.10}
Denote by $\Gamma$ the set of all progressively measurable processes
$\theta=(\theta_t)_{t\in [0,T]}$ such that
$\theta\in L^2(dt\otimes\mathbb P)$ and let
$\mathcal M$
be the set of all $\mathbb P$--martingales $M=(M_t)_{t\in [0,T)}$ which are defined on the half-open interval $[0,T)$
and satisfy
$||M||_{L^2(dt\otimes\mathbb P)}:=\mathbb E_{\mathbb P}\left[\int_{0}^T M^2_t dt\right]<\infty$.
Then,
for any $\Lambda,\alpha>0$ we have
\begin{eqnarray*}
&c(\Lambda,\alpha,\Phi_0,X)\\
&\geq \sup_{(\theta,M)\in \Gamma\times\mathcal M}\mathbb E_{\mathbb P}\left[
f\left(S_T+\sigma\int_{0}^T \theta_t dt\right)-\frac{1}{2\alpha}\int_{0}^T \theta^2_t dt\right.\\
&\left.-\Phi_0(M_0-S_0)-\frac{1}{2\Lambda}\int_{0}^T
\left|S_0+\mu t+\sigma\int_{0}^t \theta_s ds-M_t\right|^2 dt\right].
\end{eqnarray*}
\end{lemma}
\begin{proof}
Denote by $\mathcal Q$ the set of all equivalent
probability measures $\mathbb Q\sim\mathbb P$ with finite entropy
$\mathbb E_{\mathbb Q}\left[\log\left(\frac{d\mathbb Q}{d\mathbb P}\right)\right]<\infty$
relative to $\mathbb P$.
For any $\mathbb Q\in\mathcal Q$ let $\mathcal M^{\mathbb Q}$
be the set of all $\mathbb Q$--martingales $M^{\mathbb Q}=(M^{\mathbb Q}_t)_{t\in [0,T)}$ which are defined on the half-open interval $[0,T)$
and satisfy
$||M^{\mathbb Q}||_{L^2(dt\otimes\mathbb Q)}:=\mathbb E_{\mathbb Q}\left[\int_{0}^T |M^{\mathbb Q}_t|^2 dt\right]<\infty$.

From the linear growth of $f$ it follows that
$\mathbb E_{\mathbb P}\left[e^{\alpha  X}\right]<\infty$. Thus,
define the probability measure $\tilde{\mathbb P}$ by
$\frac{d\tilde{\mathbb P}}{d\mathbb P}:=\frac{e^{\alpha X}}{\mathbb E_{\mathbb P}\left[e^{\alpha  X}\right]}.$
The Cauchy–Schwarz inequality yields
that there exists $a>0$ such that
$\mathbb E_{\tilde{\mathbb P}}\left[\exp\left(a\sup_{0 \leq t\leq T} S^2_t\right)\right]<\infty.$
Hence, Assumption 2.1 in Dolinsky (2022) holds true. Thus,
by applying Theorem 2.2 in Dolinsky (2022) for the probability measure $\tilde{\mathbb P}$
and the simple equality
$$\mathbb E_{\mathbb Q}\left[\log\left(\frac{d\mathbb Q}{d\tilde{\mathbb P}}\right)\right]=
\mathbb E_{\mathbb Q}\left[\log\left(\frac{d\mathbb Q}{d\mathbb P}\right)-\alpha  X\right]+
\alpha \log\left(\mathbb E_{\mathbb P}\left[e^{\alpha  X}\right]\right) \ \ \forall \mathbb Q\in\mathcal Q$$
we obtain
\begin{eqnarray}\label{3.21}
&c(\Lambda,\alpha,\Phi_0,X)\\
&=\sup_{\mathbb Q\in\mathcal Q}\sup_{M^{\mathbb Q}\in\mathcal M^{\mathbb Q}}\mathbb E_{\mathbb Q}\left[X-\frac{1}{\alpha}\log\left(\frac{d\mathbb Q}{d\mathbb P}\right)-\Phi_0(M^{\mathbb Q}_0-S_0)-\frac{1}{2\Lambda}\int_{0}^T |M^{\mathbb Q}_t-S_t|^2 dt\right].\nonumber
\end{eqnarray}

Next,
let $C[0,T]$ be the space of continuous functions $z:[0,T]\rightarrow\mathbb R$ equipped with the uniform norm
$||z||:=\sup_{0\leq t \leq T} |z_t|$.
Denote by $\hat\Gamma\subset\Gamma$ the set of all continuous and
bounded processes $\theta=(\theta_t)_{t\in [0,T]}$ of the form $\theta=\tau(W)$ where
$\tau: C[0,T]\rightarrow C[0,T]$ is Lipschitz continuous and non-anticipative (i.e. $\tau_t(x)=\tau_t(y)$ if $x_{[0,t]}=y_{[0,t]}$).
From standard density arguments and the Lipschitz continuity of $f$
it follows that in order to complete the proof of the Lemma it is sufficient to show that
for any $(\theta,M)\in \hat\Gamma\times\mathcal M$ we have
\begin{eqnarray}\label{3.22}
&c(\Lambda,\alpha,\Phi_0,X)\nonumber\\
&\geq \mathbb E_{\mathbb P}\left[
f\left(S_T+\sigma\int_{0}^T \theta_t dt\right)-\frac{1}{2\alpha}\int_{0}^T \theta^2_t dt\right.\\
&\left.-\Phi_0(M_0-S_0)-\frac{1}{2\Lambda}\int_{0}^T
\left|S_0+\mu t+\sigma\int_{0}^t \theta_s ds-M_t\right|^2 dt\right].\nonumber
\end{eqnarray}

To this end let $(\theta,M)\in\hat\Gamma\times\mathcal M$ such that
$\theta=\tau(W)$ where $\tau$ as above.
Consider the stochastic differential equation (SDE)
\begin{equation}\label{add}
dY_t=dW_t-\tau_t(Y)dt, \ \ t\in [0,T]
\end{equation}
with the initial condition $Y_0=0$.
Theorem 2.1 from Chapter IX in Revuz \& Yor (1999)
yields that there exists a unique strong solution to the above SDE.
From the Girsanov theorem it follows that there exists a probability
measure $\mathbb Q\in\mathcal Q$ such that $W^{\mathbb Q}:=Y$ is a Brownian motion with respect to $\mathbb Q$.

From (\ref{2.bac}) and (\ref{add}) we obtain that the distribution of $(S_t)_{t\in [0,T]}$ under $\mathbb Q$ is equal to the distribution
 of $\left(S_t+\sigma \int_{0}^t\theta_s ds\right)_{t\in [0,T]}$ under $\mathbb P$.
  Moreover,
 $$\mathbb E_{\mathbb Q}\left[\frac{1}{\alpha}\log\left(\frac{d\mathbb Q}{d\mathbb P}\right)\right]=
 \mathbb E_{\mathbb Q}\left[\frac{1}{2\alpha}\int_{0}^T \tau^2_t(Y) dt\right]=
 \mathbb E_{\mathbb P}\left[\frac{1}{2\alpha}\int_{0}^T \theta^2_t dt\right].$$
Finally, choose $M^{\mathbb Q}\in\mathcal M^{\mathbb Q}$ such that
the law of $(W^{\mathbb Q}, M^{\mathbb Q})$ under $\mathbb Q$ is equal to the law of
$(W,M+\sigma W)$ under $\mathbb P$.
We conclude,
\begin{eqnarray*}
 &\mathbb E_{\mathbb Q}\left[X-\frac{1}{\alpha}\int_{0}^T\log\left(\frac{d\mathbb Q}{d\mathbb P}\right)-\Phi_0(M^{\mathbb Q}_0-S_0)-\frac{1}{2\Lambda}\int_{0}^T |M^{\mathbb Q}_t-S_t|^2 dt\right]\\
 &=\mathbb E_{\mathbb P}\left[
f\left(S_T+\sigma\int_{0}^T \theta_t dt\right)-\frac{1}{2\alpha}\int_{0}^T \theta^2_t dt\right.\\
&\left.-\Phi_0(M_0-S_0)-\frac{1}{2\Lambda}\int_{0}^T
\left|S_0+\mu t+\sigma\int_{0}^t \theta_s ds-M_t\right|^2 dt\right].
 \end{eqnarray*}
This together with (\ref{3.21}) gives (\ref{3.22}) as required.
  \end{proof}
 Next, denote by $L^2_0(\mathcal F_T,\mathbb P)$ the set of all random variables
of the form
\begin{equation}\label{3.rep}
  Z=\iota+\int_{0}^T\kappa_t dW_t
  \end{equation}
  for some  $\iota\in\mathbb R$ and a predictable and bounded process
  $\kappa=(\kappa_t)_{t\in [0,T]}$ such that for some (deterministic) $\epsilon>0$ 
  the restriction of $\kappa$ to the interval $[T-\epsilon,T]$ satisfies 
  $\kappa_{[T-\epsilon,T]}\equiv 0$.
\begin{lemma}\label{lem.11}
For any $Z\in L^2_0(\mathcal F_T,\mathbb P)$ there exists a constant $\hat C>0$ (may depend on $Z$) such that
for any $\Lambda \in (0,1)$
\begin{eqnarray*}\label{3.26}
    &\sup_{(\theta,M)\in \Gamma\times\mathcal M}\mathbb E_{\mathbb P}\left[
f\left(S_T+\sigma\int_{0}^T \theta_t dt\right)-\frac{1}{2\alpha(\Lambda)}\int_{0}^T \theta^2_t dt\right.\nonumber\\
&\left.-\Phi_0(M_0-S_0)-\frac{1}{2\Lambda}\int_{0}^T
\left|S_0+\mu t+\sigma\int_{0}^t \theta_s ds-M_t\right|^2 dt\right]\nonumber\\
&\geq\mathbb E_{\mathbb P}\left[f\left(S_0+\sigma W_T+Z\right)-\frac{\left(Z+\sigma\sqrt A\Phi_0\right)^2}{4\sigma\sqrt A}\right]+\frac{\sigma\sqrt A\Phi^2_0}{2}-\hat C\Lambda
    \end{eqnarray*}
  where, as before $\alpha(\Lambda)=\frac{A}{\Lambda}$.
\end{lemma}
\begin{proof}
Let $Z$ given by (\ref{3.rep}) and
let $\Xi$ be the map from Proposition \ref{lem5.1}. Define the deterministic function $\nu:[0,T]\rightarrow\mathbb R$ by
$\nu:=\Xi_T(\Lambda,\iota,\Phi_0)$
and for any $s<T$ define the stochastic process $(l_{\cdot,s})_{\cdot\in [s,T]}$ by
$(l_{\cdot,s})_{\cdot\in [s,T]}=\Xi_{T-s}(\Lambda, \kappa_s,0)$.

Next, introduce $(\theta,M)\in \Gamma\times \mathcal M$
 \begin{eqnarray*}
 &\theta_t:=\frac{\dot{\nu}_t-\mu}{\sigma}+\frac{1}{\sigma}\int_{0}^t \frac{\partial l_{t,s}}{\partial t} d W_s, \ \ t\in [0,T],\\
&M_t:=S_0+\frac{\int_{0}^T \nu_t dt-\Phi_0\Lambda}{T}+ \int_{0}^t \left(\sigma+\frac{1}{T-s}\int_{s}^T l_{v,s} dv\right) d W_s, \ \ t\in [0,T].
  \end{eqnarray*}
 Observe that from the definition of $\Xi$ we have
 $$\nu_0=0, \ \ \nu_T=\iota \ \  \mbox{and} \ \ l_{s,s}=0, \ \ l_{T,s}=\kappa_s \ \ \forall s.$$
 This together with the Fubini theorem, the It\^{o} Isometry,
 (\ref{2.bac}) and (\ref{3.rep}) gives
  \begin{eqnarray}\label{3.27}
&\mathbb E_{\mathbb P}\left[
f\left(S_T+\sigma\int_{0}^T \theta_t dt\right)-\frac{1}{2\alpha(\Lambda)}\int_{0}^T \theta^2_t dt
-\Phi_0(M_0-S_0)-\right.\nonumber\\
&\left.\frac{1}{2\Lambda}\int_{0}^T
\left|S_0+\mu t+\sigma\int_{0}^t \theta_s ds-M_t\right|^2 dt\right]\nonumber\\
&=\mathbb E_{\mathbb P}\left[f\left(S_0+\sigma W_T+Z\right)\right]\\
&+\frac{\mu \Lambda \iota}{\sigma^2 A}-\frac{\mu^2\Lambda }{2  \sigma^2 A}-
I(\Lambda, \nu)-\int_{0}^{T-\epsilon} \mathbb E_{\mathbb P} \left[J_s(\Lambda,l) \right]ds\nonumber
\end{eqnarray}
where
$$
I(\Lambda,\nu):=\frac{\Lambda}{ 2 \sigma^2 A}\int_{0}^T\dot{\nu}_t^2 dt+
\frac{1}{2\Lambda}\left(\int_{0}^T \nu^2_t dt-\frac{1}{T}\left(\Phi_0\Lambda-\int_{0}^T\nu_t dt\right)^2
\right)
$$
and
$$J_s(\Lambda,l):=\frac{\Lambda}{2 \sigma^2 A}\int_{s}^T \left(\frac{\partial l_{t,s}}{\partial t}\right)^2 dt+ \frac{1}{2\Lambda}
\left(\int_{s}^T l^2_{t,s} dt-\frac{1}{T-s}\left(\int_{s}^T l_{t,s} dt\right)^2\right).
$$
From Proposition \ref{lem5.1} there exists a constant $C>0$ (may depend on $\iota$ and $\kappa$) such that
\begin{equation}\label{3.28}
\left|I(\Lambda, \nu)-\frac{\left(\iota+\sigma\sqrt A\Phi_0\right)^2}{4\sigma\sqrt A}+\frac{\sigma\sqrt A\Phi^2_0}{2}\right|\leq C\Lambda
\end{equation}
and for any $s\in [0,T-\epsilon]$
\begin{equation}\label{3.29}
\left|J_s(\Lambda,l)-\frac{\kappa^2_s}{4\sigma\sqrt A}\right|\leq C\Lambda.
\end{equation}
By combining the It\^{o} Isometry and (\ref{3.27})--(\ref{3.29}) we complete the proof.
\end{proof}
We now have all the pieces in place that we need for the
\textbf{completion of the proof of Proposition \ref{prop3.1}}.
\begin{proof}
Recall the definition of $g$ given in  (\ref{def1}). From the Lipschitz continuity of $f$ it follows that there exists a bounded (measurable)
function $\zeta:\mathbb R\rightarrow\mathbb R$ such that
\begin{equation}\label{final1}
g(x)= f\left(x+\zeta(x)\right)-\frac{\zeta^2(x)}{4\sigma \sqrt A }, \ \ \ \forall x\in\mathbb R.
\end{equation}
Choose a sequence $Z_n\in L^2_0(\mathcal F_T,\mathbb P)$, $n\in\mathbb N$ such that
$$\lim_{n\rightarrow\infty} Z_n=\zeta(S_0-\sigma\sqrt A\Phi_0+ \sigma W_T)-\sigma \sqrt A\Phi_0$$
where the limit is in $L^2(\mathbb P).$
From Lemmas \ref{lem.10}--\ref{lem.11} and (\ref{final1}) we obtain
\begin{eqnarray*}
&\lim\inf_{\Lambda\downarrow 0}c\left(\Lambda,A/\Lambda,\Phi_0,X\right)\\
&\geq \sup_{n\in\mathbb N}
\mathbb E_{\mathbb P}\left[f\left(S_0+\sigma W_T+Z_n\right)-\frac{\left(Z_n+\sigma\sqrt A\Phi_0\right)^2}{4\sigma\sqrt A}\right]+\frac{\sigma\sqrt A\Phi^2_0}{2}\\
&\geq\mathbb E_{\mathbb P}\left[g\left(S_0-\sigma\sqrt A\Phi_0+\sigma W_T\right)\right]+
\frac{\sigma\sqrt A\Phi^2_0}{2}\\
&=u\left(0,S_0-\sigma\sqrt A\Phi_0\right)+\frac{\sigma\sqrt A \Phi^2_0 }{2}.
\end{eqnarray*}
\end{proof}

\section{Proof of the Upper Bound}\label{sec:4}
In order to complete the proof of Theorem \ref{thm.1} it remains to establish the following result.
\begin{proposition}\label{prop2}
Recall the trading strategies $\Phi^{\Lambda}$, $\Lambda>0$ given by (\ref{2.3}).
Then,
$$\lim\sup_{\Lambda\downarrow 0}\frac{\Lambda}{A}
 \log\left(\mathbb E_{\mathbb P}\left[\exp\left(\frac{A}{\Lambda}\left(X-V^{\Phi^{\Lambda}}_T\right)\right)\right]
\right)\leq u\left(0,S-\sigma\sqrt A\Phi_0\right)+\frac{\sigma\sqrt A \Phi^2_0 }{2}.$$
\end{proposition}
\begin{proof}
The proof will be done in three steps.\\
\textbf{Step I: }\
In this step we use the specific structure of the payoff $f$ given by
(\ref{2.form}).
Let us
show that for any $\Lambda>0$
\begin{equation}\label{4.--}
g\left(S_T-\sigma\sqrt A F^{\Lambda}_T\right)\geq f(S_T)- \frac{\sigma\sqrt A |\Phi_0  \Theta|}{\sinh\left(\sqrt{\rho(\Lambda)} T\right)}
\end{equation}
where, as before $\rho(\Lambda):=\frac{\sigma^2 A}{\Lambda^2}$.

Fix $\Lambda>0$. From (\ref{ODE})
$$\frac{d}{dt}\left[\frac{F^{\Lambda}_t}{\cosh\left(\sqrt{\rho(\Lambda)} (T-t)\right)}\right]=
\frac{\sqrt{\rho(\Lambda)}}{2 \cosh^2\left(\frac{\sqrt{\rho(\Lambda)} (T-t)}{2}\right)}\Upsilon^{\Lambda}_t, \ \ \ t\in [0,T]
$$
where, recall that
$\Upsilon^{\Lambda}_t:=\frac{\partial u}{\partial x}\left(t,S_t-\sigma\sqrt A F^{\Lambda}_t\right)$,  $t\in [0,T)$.
Clearly,
$|\Upsilon^{\Lambda}_t|\leq \Theta$, and so,
\begin{eqnarray*}
&|F^{\Lambda}_T|\leq \left|\frac{F^{\Lambda}_0}{\cosh\left(\sqrt{\rho(\Lambda)} T\right)}\right|+\Theta\int_{0}^T \frac{\sqrt{\rho(\Lambda)}}
{2 \cosh^2\left(\frac{\sqrt{\rho(\Lambda)} (T-t)}{2}\right)} dt\\
&\leq
\left|\frac{\Phi_0}{\sinh\left(\sqrt{\rho(\Lambda)} T\right)}\right|+|\Theta| .
\end{eqnarray*}
This together with (\ref{2.form}) and (\ref{def1}) gives (\ref{4.--}).\\
\textbf{Step II:}
In this step we prove that there exists a constant $\tilde C>0$ such that
\begin{equation}\label{4.0}
\left|\int_{0}^ T\Upsilon^{\Lambda}_t dt-
\int_{0}^T \Phi^{\Lambda}_t dt\right|\leq \tilde C\Lambda, \ \ \forall \Lambda>0.
\end{equation}
Fix $\Lambda>0$. From (\ref{4.0+})
\begin{eqnarray*}
&\frac{d}{dt}\left[\frac{\Phi^{\Lambda}_t}{\cosh\left(\sqrt{\rho(\Lambda)} (T-t)\right)}\right]\\
&=
\frac{\sqrt{\rho(\Lambda)}}{2 \cosh^2\left(\frac{\sqrt{\rho(\Lambda)} (T-t)}{2}\right)}
\tanh\left(\frac{\sqrt{\rho(\Lambda)} (T-t)}{2}\right)\Upsilon^{\Lambda}_t, \ \ \ t\in [0,T].
\end{eqnarray*}
We get
\begin{eqnarray*}
&\Phi^{\Lambda}_t=\Phi_0\frac{\cosh\left(\sqrt{\rho(\Lambda)}(T-t)\right)}{\cosh(\sqrt{\rho(\Lambda)} T)}
\\
&+\int_{0}^t
\frac{\sqrt{\rho(\Lambda)}\cosh\left(\sqrt {\rho(\Lambda)} (T-t)\right)}{2 \cosh^2\left(\frac{\sqrt{\rho(\Lambda)}
(T-s)}{2}\right)}\tanh\left(\frac{\sqrt{\rho(\Lambda)} (T-s)}{2}\right)\Upsilon^{\Lambda}_s ds
\end{eqnarray*}
and so, from the Fubini theorem
$$\int_{0}^T \Phi^{\Lambda}_t dt-\int_{0}^T \Upsilon^{\Lambda}_t dt=\Phi_0\frac{\tanh \left(\sqrt{\rho(\Lambda)} T\right)}{\sqrt{\rho(\Lambda)}}
-\int_{0}^T \frac{\Upsilon^{\Lambda}_s}{\cosh^2\left(\frac{\sqrt{\rho(\Lambda)} (T-s)}{2}\right)} ds.
$$
This together with the simple integral
$$\int_{0}^T \frac{ds}{\cosh^2\left(\frac{\sqrt{\rho(\Lambda)} (T-s)}{2}\right)} =
\frac{2\tanh\left(\frac{\sqrt{\rho(\Lambda)} T}{2}\right)}{\sqrt{\rho(\Lambda)}}$$
 and the inequality
$|\Upsilon^{\Lambda}_t|\leq \Theta$ gives (\ref{4.0}).
\\
\textbf{Step III:}
In this step we complete the proof.
Fix $\Lambda>0$ and introduce the process
$$M^{\Lambda}_t:=\exp\left(\frac{A}{\Lambda}\left(u\left(t,S_t-\sigma\sqrt A F^{\Lambda}_t\right)
+\frac{\sigma\sqrt A F^{\Lambda}_t\Phi^{\Lambda}_t}{2}-
V^{\Phi^{\Lambda}}_t\right)\right), \ \ t\in [0,T].$$

From the It\^{o} formula,
(\ref{2.1}), (\ref{PDE})--(\ref{2.3}) and (\ref{4.0+}) we obtain
\begin{eqnarray*}\label{4.2}
&\frac{dM^{\Lambda}_t}{M^{\Lambda}_t}=\frac{A}{\Lambda}\left(
\Upsilon^{\Lambda}_t-
\Phi^{\Lambda}_t \right) dS_t+\frac{\sigma^2 A^2}{2\Lambda^2}\left(
\Upsilon^{\Lambda}_t-
\Phi^{\Lambda}_t \right)^2 dt\nonumber\\
&-\frac{\sigma^2 A^2}{\Lambda^2}
\Upsilon^{\Lambda}_t
\left(\frac{\cosh\left(\sqrt{\rho(\Lambda)} (T-t)\right)}{2 \cosh^2\left(\frac{\sqrt{\rho(\Lambda)} (T-t)}{2}\right)}\Upsilon^{\Lambda}_t-\Phi^{\Lambda}_t\right)dt\nonumber\\
&+\frac{\sigma^2 A^2}{2\Lambda^2}\left(\tanh\left(\frac{\sqrt{\rho(\Lambda)} (T-t)}{2}\right)\Upsilon^{\Lambda}_t-
\coth\left(\sqrt{\rho(\Lambda)} (T-t)\right)\Phi^{\Lambda}_t\right)^2dt\nonumber\\
&+\frac{\sigma^2 A^2}{2\Lambda^2}\Phi^{\Lambda}_t \left(\frac{\cosh\left(\sqrt{\rho(\Lambda)} (T-t)\right)}
{2 \cosh^2\left(\frac{\sqrt{\rho(\Lambda)} (T-t)}{2}\right)}\Upsilon^{\Lambda}_t-\Phi^{\Lambda}_t\right)dt\nonumber\\
&+\frac{\sigma^2 A^2}{2\Lambda^2}\coth\left(\sqrt{\rho(\Lambda)} (T-t)\right)\Phi^{\Lambda}_t\\
&\times\left(\tanh\left(\frac{\sqrt{\rho(\Lambda)}(T-t)}{2}\right)\Upsilon^{\Lambda}_t-
\coth\left(\sqrt{\rho(\Lambda)} (T-t)\right)\Phi^{\Lambda}_t\right)\nonumber\\\\
&=\frac{A}{\Lambda}\left(\Upsilon^{\Lambda}_t-\Phi^{\Lambda}_t \right) dS_t
\end{eqnarray*}
where the last equality follows from simple calculations.

Hence, from (\ref{2.bac}) it follows that the process
$$N^{\Lambda}_t:=\exp\left(-\frac{\mu A\int_{0}^t \left(\Upsilon^{\Lambda}_t-\Phi^{\Lambda}_s\right)ds}{\Lambda}\right)M^{\Lambda}_t, \ \ t\in [0,T]$$
is a local--martingale, and so from the obvious inequality $N^{\Lambda}>0$
we conclude that this process is a super--martingale.

Finally,
\begin{eqnarray*}
&\frac{\Lambda}{A}\log\left(\mathbb E_{\mathbb P}\left[\exp\left(\frac{A}{\Lambda}\left(X-V^{\Phi^{\Lambda}}_T\right)\right)\right]\right)\\
&\leq \frac{\Lambda}{A}\log\left(\mathbb E_{\mathbb P}[M^{\Lambda}_T]\right)+\frac{\sigma\sqrt A |\Phi_0  \Theta|}{\sinh\left(\sqrt{\rho(\Lambda)} T\right)}\\
&\leq  \frac{\Lambda}{A}\log\left(\mathbb E_{\mathbb P}[N^{\Lambda}_T]\right)+\tilde C|\mu|\Lambda+\frac{\sigma\sqrt A |\Phi_0  \Theta|}{\sinh\left(\sqrt{\rho(\Lambda)} T\right)}\\
&\leq \frac{\Lambda}{A}\log \left(N^{\Lambda}_0\right)+\tilde C|\mu|\Lambda+\frac{\sigma\sqrt A |\Phi_0  \Theta|}{\sinh\left(\sqrt{\rho(\Lambda)} T\right)}\\
&=u\left(0,S_0-\sigma\sqrt A\Phi_0\coth\left(\sqrt{\rho(\Lambda)} T\right)\right)+\frac{\sigma\sqrt A\Phi^2_0 \coth \left(\sqrt{\rho(\Lambda)} T\right)}{2}\\
&+\tilde C|\mu|\Lambda+\frac{\sigma\sqrt A |\Phi_0  \Theta|}{\sinh\left(\sqrt{\rho(\Lambda)} T\right)}.
\end{eqnarray*}
The first inequality follows from (\ref{4.--})
and the relations $u(T,\cdot)=g(\cdot)$, $\Phi^{\Lambda}_T=0$.
 The second inequality is due to (\ref{4.0}).
 The super--martingale property of $N^{\Lambda}$ gives the third inequality.
 The equality is due to (\ref{2.3}).

By taking $\Lambda\downarrow 0$ we complete the proof.
\end{proof}

\section{Auxiliary Result}\label{sec:5}
For any $\mathbb T\in (0,T]$ and $x\in\mathbb R$ let $C_{0,x}[0,\mathbb T]$
be the space of all continuous functions
$z:[0,\mathbb T]\rightarrow \mathbb R$ which satisfy $z_0=0$ and $z_{\mathbb T}=x$.
 \begin{proposition}\label{lem5.1}
For any $\mathbb T\in (0,T]$ there exists a measurable map
$\Xi_{\mathbb T}:(0,1)\times \mathbb R^2\rightarrow C[0,\mathbb T)$
such that for any $\Lambda\in (0,1)$ and $x,\phi\in\mathbb R$ the continuous function
$\Xi_{\mathbb T}(\Lambda,x,\phi)\in C_{0,x}[0,\mathbb T]$ is the unique minimizer for the optimization problem
\begin{equation}\label{5.1-}
\min_{\delta\in C_{0,x}[0,\mathbb T]}\left[\frac{\Lambda}{2 \sigma^2 A}\int_{0}^{\mathbb T}\dot{\delta}_t^2 dt+
\frac{1}{2\Lambda}\left(\int_{0}^{\mathbb T}\delta^2_t dt-\frac{1}{\mathbb T}\left(\phi\Lambda-\int_{0}^{\mathbb T}\delta_t dt\right)^2\right)\right].
\end{equation}
Moreover, denote the corresponding value by $V_{\mathbb T}(\Lambda,x,\phi)$. Then, for any $\epsilon>0$ and a compact set $K\subset \mathbb R^2$ there exists a constant
$\hat C$ (may depend on $\epsilon$ and $K$) such that
\begin{equation}\label{5.1+}
\left|V_{\mathbb T}(\Lambda,x,\phi)-\frac{\left(x+\sigma\sqrt A\phi\right)^2}{4\sigma\sqrt A}+\frac{\sigma\sqrt A\phi^2}{2}\right|\leq \hat C\Lambda, \ \ \forall (\mathbb T,\Lambda, x,\phi)\in [\epsilon,T]\times (0,1)\times K.
\end{equation}
\end{proposition}
\begin{proof}
Fix $(\mathbb T,\Lambda,x,\phi)\in [\epsilon,T]\times (0,1)\times \mathbb R^2$.
 First we solve the optimization problem (\ref{5.1-}) under the additional constraint that
 $\int_{0}^\mathbb T\delta_t dt$ is given. Then, we will find the optimal
$\int_{0}^\mathbb T \delta_t dt$.

For any $y\in\mathbb R$ let $C^y_{0,x}[0,\mathbb T]\subset C_{0,x}[0,\mathbb T]$ be the subset of all functions
$\delta\in C_{0,x}[0,\mathbb T]$ which satisfy $\int_{0}^{\mathbb T}\delta_t dt=y$.
Consider the minimization problem
$$\min_{\delta\in C^y_{0,x}[0,\mathbb T]} \int_{0}^{\mathbb T} H(\dot \delta_t,\delta_t)dt$$
where $H(v_1,v_2):=\frac{\Lambda}{2\sigma^2 A}v^2_1+\frac{1}{2\Lambda} v^2_2$ for $v_1,v_2\in\mathbb R$.
This optimization problem is convex and so it has a unique solution which has
to satisfy the Euler–Lagrange equation (for details see Gelfand \& Fomin (1963))
$\frac{d}{dt}\frac{\partial H}{\partial \dot \delta_t}=\lambda+\frac{d}{dt}\frac{\partial H}{\partial \delta_t}$
for some constant $\lambda>0$ (lagrange multiplier due to the constraint $\int_{0}^{\mathbb T} \delta_t dt=y$).
Thus, the optimizer which we denote by $\hat\delta$ solves the ODE
$\ddot{\hat\delta}_t-\rho\hat\delta\equiv const$ (recall the risk-liquidity ratio 
$\rho=\rho(\Lambda):=\frac{\sigma^2 A}{\Lambda^2}$).
From the standard theory it follows that
\begin{equation}\label{sol}
\hat\delta_t=c_1\sinh(\sqrt\rho t)+c_2\sinh(\sqrt\rho(T-t))+c_3, \ \ t\in [0,\mathbb T]
\end{equation}
for some constants $c_1,c_2,c_3$.
From the three constraints
$\hat\delta_0=0$, $\hat\delta_{\mathbb T}=x$ and $\int_{0}^{\mathbb T} \hat\delta_t dt=y$ we obtain
\begin{equation}\label{sol1}
c_1=\frac{x-c_3}{\sinh(\sqrt\rho \mathbb T)}, \ \ c_2=-\frac{c_3}{\sinh(\sqrt\rho \mathbb T)} \ \
\mbox{and}  \ \ c_3=\frac{\sqrt\rho y-x\tanh(\sqrt\rho \mathbb T/2)}{\sqrt\rho \mathbb T-2\tanh(\sqrt\rho \mathbb T/2)}.
\end{equation}

We argue that
\begin{eqnarray}\label{5.1}
&\rho\int_{0}^\mathbb T \hat\delta^2_t dt+\int_{0}^{\mathbb T} \dot{\hat\delta}^2_tdt=\rho\int_{0}^{\mathbb T}
\left((\hat\delta_t-c_3)+c_3\right)^2 dt+\int_{0}^{\mathbb T} \dot{\hat\delta}^2_tdt\nonumber\\
&=\frac{\sqrt\rho}{2}\left(c^2_1+c^2_2\right)\sinh\left(2\sqrt\rho \mathbb T\right)-2c_1c_2\sqrt\rho\sinh(\sqrt\rho \mathbb T)-\rho c^2_3 \mathbb T+2\rho c_3 y\nonumber\\
&=\sqrt{\rho}x^2\coth(\sqrt\rho \mathbb T)+2\sqrt\rho c_1 c_2\sinh(\sqrt\rho \mathbb T)\left(\cosh(\sqrt\rho \mathbb T)-1\right)-\rho c^2_3 \mathbb T+2\rho c_3 y\nonumber\\
&=\sqrt{\rho}x^2\coth(\sqrt\rho \mathbb T)+\left(2\sqrt\rho \tanh(\sqrt \rho \mathbb T/2)-\rho \mathbb T\right)c^2_3\nonumber\\
&+2\left(\rho y-\sqrt\rho \tanh(\sqrt \rho \mathbb T/2)x\right)c_3\nonumber\\
&=\sqrt{\rho}\left(x^2\coth(\sqrt\rho \mathbb T)+\frac{\left(x\tanh(\sqrt\rho \mathbb T/2)-\sqrt\rho y\right)^2}{\sqrt\rho \mathbb T-2\tanh(\sqrt\rho \mathbb T/2)}\right).
\end{eqnarray}
Indeed, the first equality is obvious. The second equality follows from (\ref{sol}) and simple computations. The third equality is due
to $c_1-c_2=\frac{x}{\sinh(\sqrt\rho \mathbb T)}$. The fourth equality is due to
$c_1c_2=\frac{c^2_3-xc_3}{\sinh^2(\sqrt\rho \mathbb T)}$. The last equality follows from substituting $c_3$.

From (\ref{5.1}) we conclude that in order to minimize (\ref{5.1-}) we need to find $y$ which minimizes
the quadratic form
$$\frac{1}{2\sqrt\rho\Lambda}\
\frac{\left(x\tanh(\sqrt\rho \mathbb T/2)-\sqrt\rho y\right)^2}{\sqrt\rho \mathbb T-2\tanh(\sqrt\rho \mathbb T/2)}-\frac{1}{2\Lambda \mathbb T}\left(\phi\Lambda-y\right)^2.$$
Observe that this quadratic form is convex in $y$ and so has a unique minimum
\begin{equation}\label{sol2}
y=\frac{x\mathbb T }{2}-\frac{\phi\Lambda \left(\sqrt\rho \mathbb T-2\tanh(\sqrt\rho \mathbb T/2)\right)}{2\tanh(\sqrt\rho \mathbb T/2)}.
\end{equation}
Thus, define
$\Xi_{\mathbb T}(\Lambda,x,\phi):=\hat\delta$ where $\hat\delta$ is given by
(\ref{sol})--(\ref{sol1}) and (\ref{sol2}).
Clearly, $\Xi_{\mathbb T}(\Lambda,x,\phi)$ is the unique minimizer for (\ref{5.1-}).

Let
$$V_{\mathbb T}(\Lambda,x,\phi):=
\frac{\Lambda}{2 \sigma^2 A}\int_{0}^{\mathbb T}\dot{\hat\delta}_t^2 dt+
\frac{1}{2\Lambda}\left(\int_{0}^{\mathbb T}\hat\delta^2_t dt-\frac{1}{\mathbb T}\left(\phi\Lambda-\int_{0}^{\mathbb T}\hat\delta_t dt\right)^2\right).
$$
Finally, we prove (\ref{5.1+}). Choose $\epsilon>0$ and a compact set $K\subset\mathbb R^2$. Assume that
$(\mathbb T,x,\phi)\in [\epsilon,T]\times K$.
From (\ref{5.1}) and the equality $\rho=\frac{\sigma^2 A}{\Lambda^2}$
we get that there exists a constant $C_1$ (may depend on $\epsilon$ and $K$)
such that
\begin{equation}\label{5.100}
\left|V_{\mathbb T}(\Lambda,x,\phi)-\left(\frac{x^2}{2\sigma \sqrt A}+\frac{y^2}{\sigma\sqrt A\mathbb T^2}+\frac{\phi y}{\mathbb T}-\frac{xy}{\sigma\sqrt A\mathbb T}\right)\right|\leq C_1\Lambda
\end{equation}
where $y$ given by (\ref{sol2}).
From (\ref{sol2}) we have
$\left|y-\frac{\mathbb T}{2}\left(x-\sigma\sqrt A\phi\right)\right|\leq C_2 \Lambda$
for some constant $C_2$ (may depend on $\epsilon$ and $K$).
This together with (\ref{5.100}) gives (\ref{5.1+}) and completes the proof.
\end{proof}

\end{document}